\documentclass[journal,twoside,web]{ieeecolor}
 \usepackage{lcsys}
\usepackage{cite}
\usepackage{amsmath,amssymb,amsfonts}
\usepackage{algorithmic}
\usepackage{graphicx}
\let\labelindent\relax
\usepackage{enumitem}
\usepackage{float}
\usepackage{caption}
\usepackage{subcaption}

\usepackage{graphicx,graphics}

\newtheorem{theorem}{Theorem}
\newtheorem{corollary}{Corollary}

\newtheorem{lemma}{Lemma}
\newtheorem{definition}{Definition}

\def\BibTeX{{\rm B\kern-.05em{\sc i\kern-.025em b}\kern-.08em
    T\kern-.1667em\lower.7ex\hbox{E}\kern-.125emX}}
\begin{document}
\title{Analysis and Design of Reset Control Systems via Base Linear Scaled Graphs}
\author{T. de Groot, W.P.M.H. Heemels, and S.J.A.M. van den Eijnden
\thanks{The authors are with the Department of Mechanical Engineering, Eindhoven University of Technology, The Netherlands (email: \{t.d.groot2, m.heemels, s.j.a.m.v.d.eijnden\}@tue.nl).}
\thanks{Research partly supported by the European Union (ERC Advanced Grant, Proacthis, no. 10105538).}}

\maketitle
\thispagestyle{empty} 

\begin{abstract}
In this letter, we prove that under mild conditions, the scaled graph of a reset control system is bounded by the scaled graph of its underlying base linear system, i.e., the system without resets. Building on this new insight, we establish that the negative feedback interconnection of a linear time-invariant plant and a reset controller is stable, if the scaled graphs of the underlying base linear components are strictly separated. This result simplifies reset system analysis, as stability conditions reduce to verifying properties of linear time-invariant systems. We exploit this result to develop a systematic approach for reset control system design. Our framework also accommodates reset systems with time-regularization, which were not addressed in the context of scaled graphs before.
\end{abstract}


\section{Introduction}
\label{sec:introduction}
\IEEEPARstart{R}{eset} controllers can overcome fundamental limitations of linear time-invariant (LTI) control systems \cite{Beker01, Zhao19}. Due to the performance benefits reset control systems can offer, they have attracted significant interest over the past decades, leading to substantial progress in their analysis and design; see, e.g., \cite{Clegg58, Loon17, Dastjerdi23, Beker04, Eijnden24, Krebbekx25, Carrasco10, Nesic08}. Despite these developments, a significant challenge remains in deriving intuitive and systematic design methods for reset control systems that align with classical loop-shaping techniques for LTI systems, which still form the core of industrial control engineering today.  

To address this challenge, several works have proposed frequency-domain methods for stability analysis of reset control systems that hinge on connections with the Kalman-Yakubovich-Popov lemma \cite{Rantzer96}. Although effective, these approaches are only applicable to Lur'e-like systems, i.e., the feedback interconnection of an LTI plant and a reset element, and are often tailored toward the specific reset controller at hand \cite{Loon17, Dastjerdi23}. A different approach comes from the recently introduced notion of scaled graphs (SGs) \cite{Huang24, Chaffey23}. The SG of a nonlinear system bears a strong resemblance to the classical Nyquist diagram for LTI systems, and enables graphical Nyquist-like stability tests for general nonlinear feedback interconnections beyond the Lur'e framework. As such, SGs offer a bridge towards graphical loop-shaping methods for reset control system design.

Obtaining the SG of nonlinear systems in general, and reset systems in specific, is not straightforward, as this requires capturing system responses to an infinite (uncountable) number of input signals. Recently, computationally tractable methods for over-approximating SGs of reset systems have been proposed in \cite{Eijnden24, Groot25}. Key in these approaches is the connection between dissipativity, integral quadratic constraints (IQCs) and SGs to formulate the construction of  over-approximations of SGs as linear matrix inequalities (LMIs). These over-approximations are valuable for robust stability and performance analysis, but currently offer limited guidance for controller tuning due to an incomplete understanding in general of how changes in parameters affect the shape of the SG. For LTI systems, on the other hand, such understanding is more developed due to the strong connection between the system's Nyquist diagram and its SG, see \cite[Section V]{Chaffey23}. 

In view of the above, in this paper we aim to exploit SGs of LTI systems for the analysis and design of reset systems. To this end, we make three main contributions. First, we formally show that the SG of a reset system is bounded (in a sense that we will make precise) by the SG of its underlying base linear system (BLS), i.e., the system without resets. It thus follows that reset systems inherit input-output properties of their BLS as reflected in their SGs. This result generalizes the findings in \cite{Carrasco10}, where it was shown that a reset system is input/output/strictly passive, if its BLS is input/output/strictly passive. Second, we use this new result to formulate feedback stability conditions for reset control systems that are entirely based on the SG of the underlying BLS. As the latter can be constructed efficiently from the transfer function of the BLS \cite{Chaffey23}, (robust) stability analysis is simplified. Third, we utilize these analysis conditions for SG-based design of reset control systems. In particular, our reset control design procedure involves shaping the SG of the BLS, connecting loop-shaping with reset control design. We demonstrate the effectiveness of this new perspective through an illustrative example. Our results also allow for reset systems with time-regularization schemes \cite{Nesic08}, often crucial to guarantee well-posed behaviour of the control system. Time-regularized reset systems have not been addressed before in the SG context.

\section{Background and preliminaries}\label{sec:preliminaries}
Let $\mathbb{R}$ and $\mathbb{R}_{\geq 0}$ denote the sets of real numbers and non-negative real numbers, $\mathbb{C}$ the set of complex numbers, and $\mathbb{C}_\infty = \mathbb{C}\cup \left\{\infty\right\}$. For $z\in \mathbb{C}$, we denote its real and imaginary parts by $\textup{Re}\left\{z\right\}$ and $\textup{Im}\left\{z\right\}$, its magnitude by $|z|$, and its complex conjugate by $\bar{z}$. The inverse is defined as $z^\dagger:=\bar{z}/|z|^2$ if $z\neq 0$, and $z^\dagger = \infty$ if $z=0$. The closure of a set ${S}\subseteq \mathbb{C}_\infty$ is denoted by $\overline{{S}}$, and its complement by $S^c:=\mathbb{C}_\infty \setminus S$. A set $S\subseteq X \subseteq\mathbb{C}$ is an unbounded component of $X$ if: 1) $S$ is unbounded, 2) $S$ is connected, 3) $A\subseteq S$ and $A$ connected implies $A=S$. If $S$ is bounded, $S^c$ has one unbounded component, which we denote $(S^c)_\infty$. We refer to \cite{Conway1995} for details. The inverse of a set $S \subseteq \mathbb{C}_\infty$ is given by $S^\dagger = \left\{z^\dagger \mid z \in S\right\}$. The distance between two sets ${S}_1,S_2\subseteq \mathbb{C}_\infty$ is denoted by $\textup{dist}({S}_1,{S}_2):=\inf_{a\in {S}_1,\;b\in{S}_2}|a-b|$, where $|\infty-\infty|:=0$. 
The $n\times n$ identity matrix is denoted $I_n$. The set of $m\times m$ real symmetric matrices is denoted  $\mathbb{S}^{m}:=\{P\in\mathbb{R}^{m\times m}\mid P=P^\top\}$. We adopt standard convention for positive/negative (semi-)definite matrices and denote $\mathbb{S}_{\succeq 0}^m:=\left\{P\in \mathbb{S}^m\mid P\succeq 0\right\}$.

\subsection{Signal spaces, systems, and stability}
The set of square-integrable functions is denoted by 
$$\mathcal{L}_2^n = \left\{u: \mathbb{R}_{\geq 0}\to\mathbb{R}^n \mid \int_{0}^\infty u(t)^\top u(t) dt < \infty\right\}$$ 
and equipped with the inner product and norm $$\langle u,y\rangle = \int_{0}^\infty u(t)^\top y(t) dt \:\textup{ and }\:\|u\| = \sqrt{\langle u,u\rangle}.$$
We define the extended $\mathcal{L}_2$-space as
\begin{equation*}
      \mathcal{L}_{2e}^n = \left\{u:\mathbb{R}_{\geq 0} \to \mathbb{R}^n \mid P_Tu \in \mathcal{L}_2^n \ \ \textup{for all } T > 0\right\},
\end{equation*}
where for $T\geq 0$, $P_T$ is the truncation operator on a signal $u:\mathbb{R}_{\geq 0}\to \mathbb{R}^n$ such that $P_Tu :\mathbb{R}_{\geq 0}\to \mathbb{R}^n$ with $(P_Tu)(t)=u(t)$ when $t\in [0,T]$, and $(P_Tu)(t)=0$ when $t> T$. We write $u=0$ to mean the signal $u$ is zero almost everywhere.

We consider (possibly multi-valued) systems $H$ that map inputs $u\in \mathcal{L}_{2e}^n$ to outputs $y \in \mathcal{L}_{2e}^n$. We will only consider \emph{square} systems, i.e., systems with an equal number of inputs and outputs. Note that non-square systems can be treated by patching them with zero maps to make them square. As a shorthand notation, we write $y\in H(u)$ to indicate all possible outputs related to the input $u$ when applied to the system $H$. We assume that the system uniquely maps the zero input to zero output signals, i.e., $H(0)=\{0\}$. A system is said to be causal if $P_Ty \in P_TH(P_Tu)$ for all $T\geq 0$, all $u \in \mathcal{L}_2^n$ and all $y\in H(u)$. We will adopt the following notions of stability. 

\begin{definition}\label{def:stab}
    A causal system $H:\mathcal{L}_{2e}^n \rightrightarrows \mathcal{L}_{2e}^n$ is: 
    \begin{itemize}
        \item stable if $H(u) \in \mathcal{L}_2^n$ for all $u \in \mathcal{L}_2^n$;
        \item bounded if it is stable and 
        \begin{equation*}
       \|H\|:=\sup_{0\neq u\in \mathcal{L}_2^n} \sup_{y\in H(u)} \frac{\|y\|}{\|u\|}<\infty.
    \end{equation*}
    \end{itemize}
    We will refer to $\|H\|$ as the $\mathcal{L}_2$-gain of $H$.  
\end{definition}

A system $H$ is said to be linear if $H(a u_1+u_2) = a H(u_1)+H(u_2)$ for all $u_1,u_2\in \mathcal{L}_{2e}^n$, and $a \in \mathbb{R}$. It is said to be time-invariant if $H(S_T u)=S_T H(u)$ for all $u\in \mathcal{L}_{2e}^n$, $T\geq 0$,  and where $S_T$ is the time shift operator on a signal $u:\mathbb{R}_{\geq 0}\to \mathbb{R}^n$ such that $S_T :\mathbb{R}_{\geq 0} \to \mathbb{R}^n$ with $(S_Tu)(t)=u(t-T)$ if $t \geq T$, and $(S_T u)(t)=0$ otherwise. We consider causal linear time-invariant (LTI) systems  represented in state-space form as
\begin{equation}\label{eq:LTI}
    \mathcal{G}:\begin{cases}
        \dot{x}=Ax+Bu, \quad x(0)=0,\\
        y=Cx+Du,
    \end{cases}
\end{equation}
with state $x(t)\in \mathbb{R}^m$, input $u(t)\in \mathbb{R}^n$, output $y(t)\in \mathbb{R}^n$ at time $t \in \mathbb{R}_{\geq 0}$, and matrices $A,B,C,D$ of appropriate sizes. Equivalently, $\mathcal{G}$ can be represented in the Laplace domain (with some abuse of notation) by the transfer function
\begin{equation*}
     \mathcal{G}(s) = C(sI-A)^{-1}B+D,
\end{equation*}
with $s \in \mathbb{C}$. We call an LTI system normal, if its transfer function satisfies $\overline{ \mathcal{G}}(s) \mathcal{G}(s)= \mathcal{G}(s)\overline{ \mathcal{G}}(s)$ for al $s\in\mathbb{C}$. 

\subsection{Scaled graphs and feedback analysis}

Given signals $u\in \mathcal{L}_2^n$ and $y\in \mathcal{L}_2^n$, we define the gain $\rho(u,y)$ from $u$ to $y$ by
$
    \rho(u,y):=  \|y\|/\|u\|
$
if $u \neq 0$, $\rho(u,y)=0$ if $u,y=0$ and $\rho(u,y)=\infty$ if $u = 0, y\neq 0$. The phase $\theta(u,y)$ from $u$ to $y$ is defined by
$
    \theta(u,y):=  \arccos \frac{\langle u,y\rangle}{\|u\|\|y\|}
$
if $u,y\neq0$ and $\theta(u,y) = 0$ otherwise.
Likewise, the gain and phase from signal $u \in \mathcal{L}_{2e}^n$ to signal $y\in \mathcal{L}_{2e}^n$  for window length $T> 0$ are defined by
\begin{equation}
    \rho_T(u,y):=\!\rho(P_Tu,P_Ty) \textup{ and } \theta_T(u,y):=\!\theta(P_Tu,P_Ty).
\end{equation}
 Equipped with these notions for gain and phase, we define soft and hard scaled graphs.

\begin{definition}

    \noindent i)  For a causal and bounded system $H:\mathcal{L}_{2}^n \rightrightarrows\mathcal{L}_{2}^n$ the \emph{soft Scaled Graph} is defined as
\begin{equation*}
    \textup{SG}(H):=\big\{\rho(u,y)e^{\pm j\theta(u,y)}\mid u\in \mathcal{L}_2^n,\;y\in H(u)\big\}.
\end{equation*}
ii) For a causal system $H:\mathcal{L}_{2e}^n \rightrightarrows\mathcal{L}_{2e}^n$, the \emph{hard Scaled Graph} is defined as
\begin{equation*}
   \textup{SG}_{e}(H)\!:=\!\big\{\rho_T(u, y)e^{\pm j \theta_T(u,y)}\mid 
    \! u\! \in\! \mathcal{L}_{2e}^n, \; \!y\in \! H(u),T\!>\! 0\big\}.
\end{equation*}
\vspace{-0.4cm}
\end{definition}

%\subsection{Feedback analysis}
Scaled graphs offer an elegant means for \emph{graphical} analysis of feedback systems in Fig.~\ref{fig:FB}. Here, $H_1,H_2:\mathcal{L}_{2e}^n\rightrightarrows\mathcal{L}_{2e}^n$ are causal systems, $w \in \mathcal{L}_{2}^n$ is an external signal, and $u_1, u_2, y_1, y_2$ are internal signals. We denote the closed-loop mapping from $w$ to $y_1$ by $\Sigma$. We call this system well-posed if for all $w \in \mathcal{L}_2^n$, there exist $u_1, u_2 \in \mathcal{L}_{2e}^n$ satisfying $u_1\in \{w\}-H_2(u_2)=\left\{w-y_2 \mid y_2\in H_2(u_2)\right\}$ and $u_2\in H_1(u_1)$.

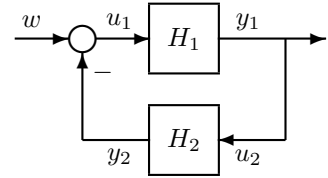
\begin{figure}[htbt!]
\centering
\setlength{\unitlength}{0.9mm}
\begin{picture}(50,25)
\thicklines \put(0,20){\vector(1,0){8}} \put(10,20){\circle{4}}
\put(12,20){\vector(1,0){8}} \put(20,15){\framebox(10,10){$H_1$}}
\put(30,20){\vector(1,0){16}} \put(40,20){\line(0,-1){15}}
\put(40,5){\vector(-1,0){10}}  
 \put(20,0){\framebox(10,10){$H_2$}}
\put(20,5){\line(-1,0){10}} \put(10,5){\vector(0,1){13}}
\put(13,0){\makebox(5,5){$y_2$}} \put(32,20){\makebox(5,5){$y_1$}}
\put(0,20){\makebox(5,5){$w$}}  
\put(13,20){\makebox(5,5){$u_1$}} \put(32,0){\makebox(5,5){$u_2$}}
\put(10,10){\makebox(6,10){$-$}}
\end{picture} 
  \caption{Feedback interconnection $\Sigma$.}
        \label{fig:FB}
\end{figure}  

\begin{theorem}[\cite{Chen25}]\label{thm:stab-combined}
Let $H_1, H_2 : \mathcal{L}_{2e}^n \rightrightarrows \mathcal{L}_{2e}^n$ be causal, stable systems in negative feedback interconnection. Suppose the interconnection of $H_1$ and $\mu H_2$ is well-posed for all $\mu \in (0,1]$. If there exists $r > 0$ such that for all $\mu \in (0,1]$,
\begin{equation}
     \textup{dist}(\textup{SG}^\dagger(-H_1), \textup{SG}(\mu H_2)) \geq r,
\end{equation}
     then $\Sigma$ is stable. If, in addition, one or both of $\textup{SG}(H_1)$ and $\textup{SG}(H_2)$ satisfy the chord property\footnote{The SG of a system $H$ is said to satisfy the chord property if, for each $z \in \textup{SG}(H)$, $\lambda z+(1-\lambda )\bar{z} \in \textup{SG}(H)$ with $\lambda \in [0,1]$, see \cite{Chaffey23}.}, then $\|\Sigma\| \leq 1/r$. \hfill $\square$
\end{theorem}

\noindent A few comments regarding Theorem~\ref{thm:stab-combined} are in order:
\begin{enumerate}
    \item The chord property enables performance guarantees for the closed-loop system $\Sigma$ in terms of its $\mathcal{L}_2$-gain. It distinguishes conditions for stability from those for performance, as the former do not require the chord property.
    \item An analogous result can be formulated using hard SGs. In this case, unbounded systems are allowed, and the homotopy condition ``for all $\mu\in(0,1]$'' can be removed ($\mu=1$ is sufficient), see \cite{Chen25} for details on hard SGs. 
\end{enumerate}

\section{SGs and reset systems}\label{sec:SRG_and_reset}

\subsection{Reset systems}
The class of systems that is considered in this paper is formed by reset systems with time-regularization given by
\begin{equation}\label{eq:RT}
    \mathcal{R}:\begin{cases}
        \:\:\:\dot{x} = Ax+Bu,\:\: \dot{\tau} = 1, &\textup{ if } q\in \mathcal{F} \textup{ or } \tau \leq \delta, \\
        x^+ = Rx, \quad \quad\:\, \tau^+=0, &\textup{ if } q\in \mathcal{J} \textup{ and } \tau \geq \delta, \\
        \:\:\:y =Cx+Du, 
    \end{cases}
\end{equation}

with states $x(t)\in\mathbb{R}^m$, $\tau(t) \geq 0$, $x(0)=0$, $\tau(0)=0$, threshold $\delta >0$, input $u(t)\in\mathbb{R}^n$, output $y(t)\in \mathbb{R}^n$, and $q(t) = [x^\top(t), u{^\top}(t)]^\top \in \mathbb{R}^{m+n}$ all at time $t \in \mathbb{R}_{\geq 0}$, and where $x^+(t)= \lim_{s \downarrow t}x(s)$. The matrix $R \in \mathbb{R}^{m\times m}$ defines the reset map, and the jump and flow sets are given by
\begin{subequations}\label{eq:FJ}
\begin{align}
	\mathcal{F} &= \{v \in \mathbb{R}^{m+n} \;|\; v^TMv\geq 0\} \label{eq:canonical_flow_set_reset},\\
	\mathcal{J} &= \{v \in \mathbb{R}^{m+n} \;|\; v^TMv\leq 0\} \label{eq:canonical_jump_set_reset},
\end{align}
\end{subequations}
with $M\in\mathbb{S}^{m+n}$. Variable $\tau$ is a ``timer'' that ensures that the reset times, denoted by $t_0,t_1,\dots$ satisfy $t_{i+1}-t_i\geq \delta$ for all $i=0,1,\ldots$ and solutions in the sense of \cite{Heemels26} are well-defined.

We associate with $\mathcal{R}$ in \eqref{eq:RT} a base linear system (BLS) given by the tuple $(A,B,C,D)$, i.e., an LTI system as in \eqref{eq:LTI} without resets (and timer). We denote this system by $\mathcal{R}_{\textup{BLS}}$. As we will show, certain reset systems inherit structural properties from the SG of their underlying BLS. 

\subsection{SGs for LTI systems}
Before connecting reset SGs to their BLS, we require some additional results for LTI systems. 

\begin{lemma}[\cite{Groot25}]
	\label{th:LTI_approx}
Consider an LTI system $H$ of the form (\ref{eq:LTI}) and assume that $A$ is Hurwitz. Let 
\begin{align} \label{eq:thetaPi}
\Theta(\Pi) = \begin{bmatrix}
		C & D \\ 0 & I_n
	\end{bmatrix}^\top(\Pi\otimes I_n)\begin{bmatrix}
		C & D \\ 0 & I_n
	\end{bmatrix}, 
\end{align} 
with $\Pi \in \mathbb{S}^2$ and $\det(\Pi)<0$, and where $\otimes$ is the Kronecker product. Define the region $\mathcal{S}(\Pi)$ as 
\begin{equation}\label{eq:SP}
    \mathcal{S}(\Pi) = \left\{z \in \mathbb{C} \mid \begin{bmatrix}
        \bar{z} \\1
    \end{bmatrix}^\top\Pi\begin{bmatrix}
        z \\1
    \end{bmatrix} \geq 0\right\}.
\end{equation}
If there exists a matrix $P\in \mathbb{S}^m$ that satisfies the LMI
	\begin{align} \label{eq:KYP_LMI}
		 \begin{bmatrix}
			A^\top P+PA &PB \\
            B^\top P & 0\end{bmatrix} - \Theta(\Pi) &\preceq 0, 
    \end{align} 
then $\textup{SG}(H) \subset \mathcal{S}(\Pi)$. If $P\in\mathbb{S}_{\succeq 0}^m$, then $\textup{SG}_e(H) \subset \mathcal{S}(\Pi)$. \hfill $\square$
\end{lemma}

Lemma~\ref{th:LTI_approx} provides an algorithmic way for computing over-approximations $\mathcal{S}(\Pi)$ in \eqref{eq:SP} of $\textup{SG}(H)$ and $\textup{SG}_e(H)$. These over-approximations are characterized by the matrix $\Pi$, for which a parametrization is suggested in \cite{Groot25} as
\begin{equation}\label{eq:PIS}
    \Pi(\sigma, \lambda, r) = \sigma\begin{bmatrix}
        1 & -\lambda \\
        -\lambda & \lambda^2-r^2
    \end{bmatrix}, 
\end{equation}
with $\sigma \in \left\{-1, +1\right\}$, $\lambda \in \Lambda \subseteq\mathbb{R}$, and $r >0$. Note that $\det( \Pi (\sigma, \lambda, r)) = -r^2<0$, and $\mathcal{S}(\Pi(-1,\lambda,r))$ is the \emph{interior} of a disk in the complex plane centered at $\lambda$ with radius $r$, and $\mathcal{S}(\Pi(+1,\lambda,r))$ is the \emph{exterior} of such a disk. 

Taking intersections of these regions, we find from application of Lemma~\ref{th:LTI_approx} the SG over-approximations
\begin{align} \label{eq:over_app}
\textup{SG}(H) \subseteq \!\!\!\!\!\! \bigcap_{\Pi \in \mathbf{\Pi}(H,\Lambda)}\!\!\!\!\!\! \mathcal{S}(\Pi),\:\:\text{ and } \:\:\textup{SG}_e(H) \subseteq \!\!\!\!\!\! \bigcap_{\Pi \in \mathbf{\Pi}_{\succeq 0}(H,\Lambda)} \!\!\!\!\!\! \mathcal{S}(\Pi)
\end{align} 
for all $\Lambda \subset \mathbb{R}$, where
\begin{equation}\label{eq:Pi0}
 \begin{split} &  \mathbf{\Pi}(H,\Lambda)
=
\left\{
\Pi \text{ in } \eqref{eq:PIS}\mid \lambda \in \Lambda, \sigma \in \left\{-1, +1\right\}, \right.\\
&\quad\qquad \qquad \qquad \left.
r>0, \exists\, P \in \mathbb{S}^m \text{ satisfying } \eqref{eq:KYP_LMI}
\right\},
\end{split}
\end{equation}
and
\begin{equation} \label{eq:Pigeq0}
    \begin{split}
        & \mathbf{\Pi}_{\succeq 0}(H,\Lambda)  \!=\!
\left\{
\Pi \text{ in } \eqref{eq:PIS}\mid  \lambda \in \Lambda, \sigma \in \left\{-1, +1\right\}, \right. \\
&\qquad \qquad \qquad \qquad \left.r>0, \exists P \! \in \mathbb{S}_{\succeq0}^m \text{ satisfying } \eqref{eq:KYP_LMI}
\right\}.
\end{split}
\end{equation}
It follows from \cite[Theorem 12]{Groot25} that when $\Lambda = \mathbb{R}$, the first inclusion in \eqref{eq:over_app} reduces to an equality for the closure of the SG of normal, controllable LTI systems. For further details on computational aspects, the reader is referred to \cite{Groot25}.

 \subsection{Properties inherited from base linear SGs}
We will now work toward our main result: connecting the SG of a reset system in \eqref{eq:RT} to that of its BLS. To this end, we first introduce the concept of a patched set, needed to formalize the notion of one SG being bounded by another SG. 
\begin{definition}\label{def:patch}
    Let $S\subseteq \mathbb{C}_\infty$ be closed and bounded. Then,
    \begin{equation}
      \operatorname{patch}({S})
=
{S} \cup
\left(S^c \setminus (S^c)_{\infty}\right).
    \end{equation}
\end{definition}
\vspace{0.2cm}

Intuitively, in the context of soft and hard SGs, the patch of a set in Definition~\ref{def:patch} serves to define a ``filled'' SG that keeps non-convexities of the SG's outer shape. An example of $\textup{SG}(H)$ and $ \operatorname{patch}(\textup{SG}(H))$ is shown in Fig.~\ref{fig:hulla}.

\begin{figure}[b]
	\centering
	\includegraphics[width = 0.7\linewidth]{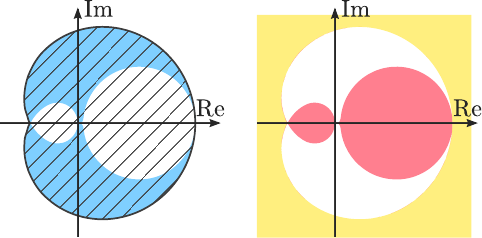}
	\caption{Left: $\textup{SG}$ (blue) and  $\textup{patch}(\textup{SG})$ (hatched region). Right: $(\textup{SG}^c)_{\infty}$ (yellow) and $\left(\textup{SG}^c \setminus (\textup{SG}^c)_{\infty}\right)$ (red).}
	\label{fig:hulla}
\end{figure}

We next state the main theorem of this section, linking the SG of a reset system to the SG of its BLS.

\begin{theorem}\label{th:reset_1}
    Consider the reset system $\mathcal{R}$ in \eqref{eq:RT}, with a base linear system $\mathcal{R}_\textup{BLS}$ that is normal, $A$ Hurwitz and $(A,B)$ controllable. Suppose that for each $P \in \mathbf{P}(\mathcal{R}_{\textup{BLS}}, \mathbb{R})$ with 
    \begin{equation*}
    \mathbf{P}(H,\Lambda)\! = \! \left\{\!P\!\in\! \mathbb{S}_{\succeq 0}^m \!\mid \exists \Pi  \in  \mathbf{\Pi}_{\succeq 0}(H,\Lambda) \textup{ satisfying }\eqref{eq:KYP_LMI}\right\},
\end{equation*}   
    there exists $\rho \geq 0$ so that the reset map $R$ satisfies 
        \begin{equation}\label{eq:RPR}
       \begin{bmatrix} P & 0 & R^\top P^{1/2}\\
       0 & 0 & 0 \\
       P^{1/2} R & 0 & 0\end{bmatrix}+\begin{bmatrix}
           \rho M & 0 \\
           0 & I_m
       \end{bmatrix}\succeq 0.
    \end{equation}
Then, the following inclusions hold:
    \begin{align*}
        \textup{SG}(\mathcal{R}) \subseteq \textup{patch}(\overline{\textup{SG}}(\mathcal{R}_{\textup{BLS}})), \;
    \textup{SG}_e(\mathcal{R})\subseteq \textup{patch}(\overline{\textup{SG}}(\mathcal{R}_{\textup{BLS}})). 
    \end{align*}
    \end{theorem}
    \vspace{0.2cm}
    
The proof of Theorem~\ref{th:reset_1} is provided in the Appendix.

\begin{corollary}\label{col:reset_1}
    The inclusions in Theorem \ref{th:reset_1} hold true for any reset system of the form \eqref{eq:RT} with $A$ Hurwitz, $(A,B)$ controllable, and $R=\alpha I$, $\alpha \in [-1,1]$. 
\end{corollary}

\begin{proof}
    Take a Schur complement of \eqref{eq:RPR}. Choose $\rho=0$. The upper-left corner of the resulting matrix reduces to $(\alpha^2-1)P$. Since $P\succeq 0$ and $\alpha^2-1 \leq 0$, \eqref{eq:RPR} is  satisfied. 
\end{proof}

We remark that for reset systems with reset maps as in Corollary \ref{col:reset_1}, $\textup{patch}(\textup{SG}(\mathcal{R}_{BLS}))$ can be obtained without the need to solve LMIs. As $\textup{SG}(\mathcal{R}_{BLS})$ can be directly constructed from its transfer function, see, e.g., \cite{Chaffey23}, and $\textup{patch}(\textup{SG}(\mathcal{R}_{BLS}))$ is graphically intuitive to construct. Furthermore, $M$ can be freely chosen to optimize for performance.

The results in Theorem~\ref{th:reset_1} and Corollary~\ref{col:reset_1} reveal that certain reset systems inherit key properties of their BLS, in terms of their scaled graphs. These results, in particular Corollary~\ref{col:reset_1}, are a generalization of \cite[Proposition 1]{Carrasco10} which shows that a reset system with $R=0$ is input/output/strictly passive, if its BLS is input/output/strictly passive. The reset map structure considered in Corollary~\ref{col:reset_1} captures many known reset elements found in the literature such as first-order reset element (FORE) \cite{Nesic08}, and partial reset compensator \cite{Carrasco10, Dastjerdi23}. In \cite[Proposition 2]{Carrasco10} passivity inheritance from the BLS was shown to also extend to reset systems with $R\neq 0$, provided that a condition akin to \eqref{eq:RPR} holds. 

In the remainder, we refer to reset systems in \eqref{eq:RT} having a base linear system $\mathcal{R}_{\textup{BLS}}$ with $A$ Hurwitz, $(A,B)$ controllable, and $R,M$ satisfying \eqref{eq:RPR}, as \emph{admissible reset systems}.

\section{Stability analysis and design}\label{sec:analysis_design}

\subsection{Conditions for stability}
We present a simple stability test for a (nonlinear) feedback system containing an admissible reset element. 

\begin{theorem}\label{th:intres}
    Consider the feedback interconnection $\Sigma$ in Fig.~\ref{fig:FB}, with $H_1$ a causal, stable system, and $H_2$ an admissible reset system $\mathcal{R}$ with associated BLS $\mathcal{R}_{\textup{BLS}}$. Suppose the interconnection of $H_1$ and $\mu \mathcal{R}$ is well-posed for all $\mu \in (0,1]$. If there exists $r > 0$ such that for all $\mu \in (0,1]$,    
 \begin{equation}\label{eq:SGRESET2}
     \textup{dist}(\textup{SG}^\dagger(-H_1), \textup{patch}(\overline{{\textup{SG}}}(\mu\mathcal{R}_{\textup{BLS}}))) \geq r,
\end{equation}
then $\Sigma$ is stable. If, in addition, $\textup{SG}(H_1)$ or $\textup{patch}({\overline{\textup{SG}}}(\mathcal{R}_{\textup{BLS}}))$ or both satisfy the chord property, then $\|\Sigma\|\leq 1/r$. \hfill $\square$
\end{theorem}

\begin{proof}
   Since  $\mathcal{R}$ is an admissible reset system, Theorem~\ref{th:reset_1} applies and thus $\textup{SG}(\mathcal{R}) \subseteq\textup{patch}({\overline{\textup{SG}}}(\mathcal{R}_{\textup{BLS}}))$. The result then follows from applying Theorem~\ref{thm:stab-combined}. 
\end{proof}

We make the following remarks on Theorem~\ref{th:intres}:
\begin{enumerate}
\item As for Theorem~\ref{thm:stab-combined}, a version of Theorem~\ref{th:intres} can be formulated with $\textup{SG}_e^\dagger(-H_1)$, allowing for, e.g., integrators. 
\item When $H_1$ is LTI, Theorem~\ref{th:intres} shows that the closed-loop reset  system is stable if the simpler underlying closed-loop  LTI system is stable according to SG analysis. 
\item In \eqref{eq:SGRESET2}, $\textup{patch}({\overline{\textup{SG}}}(\mathcal{R}_{\textup{BLS}}))$ can be replaced by its over-approximation $\bigcap_{\Pi \in \mathbf{\Pi}_{\succeq 0}(\mathcal{R}_{\textup{BLS}},\Lambda)}\mathcal{S}(\Pi)$ with $\Lambda$ containing a finite number of points.
\item The conditions in Theorem~\ref{th:intres} are {independent of $M$} and $\delta$, and pose limited structure on $R$. This is in contrast with frequency-domain results such as the circle-criterion \cite{Loon17} and the $H_\beta$-condition \cite{Beker04}, where the structure of $M$ and $R$ is exploited. In the circle-criterion, $M$ represents a sector, while $R=0$ is key for the $H_\beta$-condition. 
\end{enumerate}

\subsection{Design for stability}\label{sec:design}
We exploit the analysis conditions in Theorem~\ref{th:intres} in a procedure to design a reset controller that stabilizes the loop in Fig.~\ref{fig:FB} (with $H_1$ the plant, and $H_2=\mathcal{R}$ the reset controller). We assume that either $\textup{SG}(H_1)$ or $\textup{SG}_e(H_1)$ (or both) are given. 
 
\begin{enumerate}[label=\textbf{Step~\arabic*:}, leftmargin=*, align=left]
\item Design a base LTI controller with transfer function $\mathcal{R}_{\textup{BLS}}(s) = C(sI-A)^{-1}B+D$ that satisfies \eqref{eq:SGRESET2}.

\item Choose $\Lambda$ and obtain an over-approximation of $\textup{patch}(\overline{\textup{SG}}(\mathcal{R}_{\textup{BLS}}))$ by solving \eqref{eq:KYP_LMI} for each $\lambda \in \Lambda$; collect all corresponding matrices $P$ in $\mathbf{P}(\mathcal{R}_{\textup{BLS}},\Lambda)$. 

\item Solve \eqref{eq:RPR} for all $P\in \mathbf{P}(\mathcal{R}_{\textup{BLS}},\Lambda)$ with common decision variables $\rho$, $R$ and $M$. 
\end{enumerate}
The above steps lead to an admissible reset controller $\mathcal{R} = (A,B,C,D,R,M,\delta)$, where $\delta\geq 0$ can be picked arbitrary, that satisfies the conditions of Theorem~\ref{th:intres} and thus results in a stable closed-loop design. We highlight the following points:
\begin{enumerate}
    \item Design of the BLS in \textbf{Step 1} can be done via, e.g., loop-shaping \cite{Skogestad05}. For example, lead filters can rotate $\textup{SG}(\mathcal{R}_{\textup{BLS}})$ to the right half complex plane, and a proportional gain can scale $\textup{SG}(\mathcal{R}_{\textup{BLS}})$ to an appropriate size. 
    \item The choice for $\Lambda$ in \textbf{Step 2} determines the accuracy of the over-approximation of $\operatorname{patch}(\overline{\textup{SG}}(\mathcal{R}_{\textup{BLS}}))$.
    \item The reset structure is synthesized in \textbf{Step 3}. In \eqref{eq:RPR}, structure can be given to $R$ and $M$. For example, partial resetting can be enforced via, e.g., $R = \textup{diag}(R_{11}, I_{m-p})$, with $R_{11}$ a free matrix. The structure of $M$ for the single-input single-output case can, for example, be fixed to \begin{equation}
       M=\begin{bmatrix}
           C & D\\
           0 & I_n
       \end{bmatrix}^\top \begin{bmatrix}
            -k_1 & 1 \\
            1 & k_2
        \end{bmatrix}\begin{bmatrix}
            C & D\\
            0 & I_n
        \end{bmatrix} \label{eq:M_structure}
    \end{equation} with $k_1>k_2$, which induces a conic partition of the input-output space of a reset controller. 
\end{enumerate}

\section{Example}
Consider the negative feedback interconnection in Fig.~\ref{fig:FB} where $H_1$ is an LTI system with transfer function $H_1(s) = 1/(s(s+0.2))$ describing a moving mass with friction, and $H_2$ a to-be-designed reset controller. Note that $H_1$ contains an integrator and maps $\mathcal{L}_{2e}$ to $\mathcal{L}_{2e}$.  Following \textbf{Step 1}, we design $\mathcal{R}_{\textup{BLS}}(s)= (0.055/({s^2+1 s+1}))+0.1$ to satisfy \eqref{eq:SGRESET2}, see  Fig.~\ref{fig:example_1} (left). From \textbf{Step 2} we obtain $\mathbf{P}(\mathcal{R}_{\textup{BLS}},\Lambda)$ with $\Lambda =\left\{-1,-1.01,\dots,1\right\}$. For \textbf{Step 3}, we consider $R$ a free variable and let $M$ be as in \eqref{eq:M_structure}. We iterate over $k_1$ and set $k_2=0$. This results in combinations $(R, k_1)$ that satisfy \eqref{eq:RPR}, which offers a parameter space that we can explore for performance with stability guaranteed. We run closed-loop simulations for all feasible $(R,k_1)$ and select the pair that minimizes $\|u_1\|$ in response to a step input $w$, resulting in $R=0$ and $k_1=6.2$. 

    Step-responses of both the BLS and reset control system are shown in Fig.~\ref{fig:example_1} (right). Overshoot and settling time are reduced, illustrating the potential of our design procedure.
    \begin{figure}[h]
	\centering
	\begin{subfigure}[t]{0.4\linewidth}
\centering
        \includegraphics[height=3.3cm]{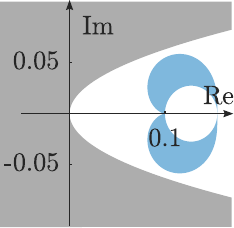}

         \label{fig:5}
\end{subfigure}
\begin{subfigure}[t]{0.55\linewidth}
\centering
    \includegraphics[height=3.3cm]{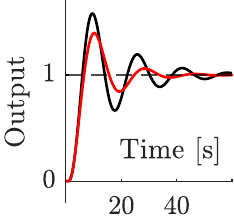}
    \label{fig:_6}
    \end{subfigure}
	\caption{Left: $\textup{SG}_e^\dagger(-H_1)$ (grey) and ${\textup{SG}}(\mathcal{R}_{\textup{BLS}})$ (blue). Right: step-response $y_1$ of BLS (black) and reset system (red).}
	\label{fig:example_1}
\end{figure}

\section{Conclusions}\label{sec:conclusions}
We established that a broad class of reset control systems inherit scaled graph properties from their base linear system, i.e., the underlying linear system without resets. This result is exploited  for the analysis and design of closed-loop reset control systems with (and without) time-regularization. Due to this key result, closed-loop stability analysis reduces to verifying scaled graph properties of the base linear system, and design involves shaping the base linear controller characteristics. The results in this paper provide new possibilities and tools in loop-shaping techniques for reset control design.  

\appendices

\section*{Appendix: Proof of Theorem \ref{th:reset_1}}
\noindent Let $H$ denote a normal LTI system as in \eqref{eq:LTI} with $A$ Hurwitz and $(A,B)$ controllable. By definition of $\operatorname{patch}(\overline{\textup{SG}}(H))$, we have to show that
   \begin{align*} 
\overline{\textup{SG}}(H)\cup \left((\overline{\textup{SG}}(H))^c\setminus ((\overline{\textup{SG}}(H))^c)_{\infty}\right)=\!\!\!\!\!\bigcap_{\Pi \in \mathbf{\Pi}_{\succeq 0}(H,\mathbb{R})} \mathcal{S}(\Pi),
\end{align*} 
which is well defined, since $H$ is bounded and thus $\overline{\textup{SG}}(H)$ is bounded. First, recall from \cite[Theorem 12]{Groot25} the equality \begin{equation}\label{eq:SG0}\overline{\textup{SG}}(H) = \bigcap_{\Pi \in \mathbf{\Pi}(H,\mathbb{R})}\mathcal{S}(\Pi),\end{equation}
with $\mathcal{S}(\Pi)$ in \eqref{eq:SP} and $\mathbf{\Pi}(H,\mathbb{R})$ in \eqref{eq:PIS}. 

Next, we will make a decomposition of $\overline{\textup{SG}}(H)$ based on regions characterized by matrices $\Pi\in \mathbf{\Pi}(H,\mathbb{R})$ with specific properties. We consider two cases: $\Pi\in \mathbf{\Pi}(H,\mathbb{R})$ with $\sigma=-1$ (i.e., $\Pi$ describes the interior of a disk) and $\Pi\in \mathbf{\Pi}(H,\mathbb{R})$ with $\sigma = +1$ (i.e., $\Pi$ describes the exterior of a disk).

\textbf{Case 1: $\Pi \in \mathbf{\Pi}(H,\mathbb{R})$ with $\sigma = -1$.} In this case, it follows that \eqref{eq:KYP_LMI} implies $A^\top P+PA +C^\top C \preceq 0$, which, since $C^\top C\succeq 0$, implies $A^\top P+PA\preceq 0$. Since, $A$ is Hurwitz, it follows that $P\succeq 0$. Therefore, those $\Pi \in \mathbf{\Pi}(H,\mathbb{R})$ with $\sigma = -1$  belong to $\mathbf{\Pi}_{\succeq 0}(H,\mathbb{R})$. We collect these matrices in the set
\begin{equation}
    \Sigma^{-1}_{\succeq 0}(H):=\left\{\Pi \in \mathbf{\Pi}(H,\mathbb{R})\mid \sigma = -1\right\},
\end{equation}
and define the corresponding region
\begin{equation}
    S_1(H) = \bigcap_{\Pi \in \Sigma_{\succeq 0}^{-1}(H)}\mathcal{S}(\Pi).
\end{equation}

\textbf{Case 2: $\Pi \!\in \!\mathbf{\Pi}(H,\!\mathbb{R})$ with $\sigma\! = \!+1$.} In this case, \eqref{eq:KYP_LMI} implies
\begin{equation}\label{eq:LMI2}
    \begin{bmatrix}
        A^\top P+PA-C^\top C & PB-C^\top(D-\lambda I) \\
        B^\top P-(D-\lambda I)^\top C & F
    \end{bmatrix}\preceq 0,
\end{equation}
with $F = -(D-\lambda I)^\top (D-\lambda I)+r^2I$. Since $r>0$, $F\preceq 0$ implies $(D-\lambda I)^\top (D-\lambda I)\succ 0$ and thus the matrix $D-\lambda I$ is non-singular. Taking the Schur complement of \eqref{eq:LMI2} (dropping $r^2I\succ 0$ in the left-hand side of \eqref{eq:LMI2}) leads to
\begin{equation}\label{eq:SCHUR}
    A^\top P+PA-C^\top C +V W^{-1}V^\top \preceq 0,
\end{equation}
where $V = PB-C^\top (D-\lambda I)$ and $W  = (D-\lambda I)^\top (D-\lambda I)$. We then have
\begin{equation*}
    -C^\top C+VW^{-1}V^\top = 
    PBW^{-1}B^\top P - \textup{He}(PB(D-\lambda I)^{-1}C),
\end{equation*}
where $\textup{He}(M):=M+M^\top$, and we have used $M(M^\top M)^{-1}M^\top = I$ for any square and invertible $M$. Since $W^{-1}\succ 0$, we find $PBW^{-1}B^\top P\succeq 0$ such that \eqref{eq:SCHUR} implies
\begin{equation}\label{eq:XP}
    X^\top P+PX \preceq 0,
\end{equation}
with $X = A-B(D-\lambda I)^{-1} C$. For $\lambda \neq 0$, $X$ can be regarded as the system matrix that arises when applying  the feedback $u = \tfrac{1}{\lambda}y$ to the LTI system $\dot{x}=Ax+Bu$, $y=Cx+Du$. The case $\lambda \to  0$ can be interpreted as the limit of high-gain feedback, where the gain $\tfrac{1}{\lambda} \to \pm \infty$. Using this feedback perspective, the closed-loop system dynamics read $\dot{x}=Xx=(A-\tfrac{1}{\lambda}B(-\tfrac{1}{\lambda}D+I)^{-1}C)x$. Hence, when $X$ is Hurwitz, \eqref{eq:XP} implies $P\succeq0$. Next, we find those $\lambda$ for which this is true. 

Using the previous perspective and the fact that $A$ is Hurwitz, it follows from the generalized Nyquist criterion \cite{Skogestad05} that $X$ is Hurwitz if and only if $\det(-\tfrac{1}{\lambda} H(s)+I)\neq 0$ and does not encircle the origin, where $H(s)=C(sI-A)^{-1}B+D$. This is equivalent to requiring the characteristic loci of $H(s)$ to not pass through nor encircle the point $\lambda+0j$. Note that the characteristic loci correspond to the spectrum $\rho(H)$ defined as
      \begin{equation*}
            \rho(H) = \!\!\!\!\!\!\bigcup_{s \in \mathbb{R}\cup\{\infty\}}\!\left\{ \alpha \!\in\! \mathbb{C} \mid  \det\!\left(\lim_{\omega\rightarrow s}\!H(j\omega)-\alpha I\right)\!=\!0,\omega\!\in\!\mathbb{R} \right\},
    \end{equation*}
which is bounded since $H$ is stable. Furthermore, it holds that $\overline{\textup{SG}}(H)\cap \mathbb{R} = \rho(H)\cap \mathbb{R}$, which follows from hyperbolic convexity of $\textup{SG}(H)$, see \cite{Chaffey23}. Therefore, the critical points for stability are $p_0:= \min\{\rho(H)\cap\mathbb{R}\}$ and $p_1:=\max\{\rho(H)\cap\mathbb{R}\}$.

Suppose $\lambda < 0$. In that case stability requires $-p_0/\lambda > -1$, leading to $p_0 >\lambda$. For $\lambda >0$ we need $-p_1 /\lambda >-1$, leading to $p_1 <\lambda$. Note that for $\lambda = 0$, we have $\tfrac{1}{\lambda}\to\pm \infty$, so that $X$ is Hurwitz if $p_0 >0$ (i.e., $H$ has infinite positive gain margin) or $p_1 <0$ (i.e., $H$ has infinite negative gain margin). Let $\Lambda := [p_0,p_1]$. Then, $X$ is Hurwitz for all $\lambda \in \mathbb{R}\setminus \Lambda$, and, therefore, \eqref{eq:XP} implies $P\succeq0$. Thus, those $\Pi\in \mathbf{\Pi}(H,\mathbb{R})$ with $\sigma = +1$ and $\lambda \in \mathbb{R}\setminus \Lambda$ also belong to $\mathbf{\Pi}_{\succeq0}(H,\mathbb{R})$.

% ==========================================
% ================= New proof bit ===========
% ==============================================
In the case that $\lambda \in \Lambda$, $X$ is not Hurwitz, and thus \eqref{eq:XP} does not directly imply anything on $P$. However, we can exclude $P\succeq 0$ using the following argument by contradiction. First, we note the following. Take $y \in \mathbb{C}^n \neq 0$ such that
\begin{equation}
    y^H (X^\top P+PX)y = 0,
\end{equation}
with $y^H$ the hermitian transpose. We have
\begin{equation}
    y^H (A^\top P+PA)y -y^\top \textup{He}(PB(D-\lambda I)^{-1}C)y = 0.
\end{equation}
Via \eqref{eq:SCHUR} we find $y^H PBW^{-1}B^\top Py \leq 0$. Since $W^{-1}\succ 0$ this yields $y^H PBW^{-1}B^\top Py = 0$ and $B^\top Py = 0$. We then have $ y^H (A^\top P+PA-C^\top C)y = -y^H VW^{-1}V^\top y$.
Now consider \eqref{eq:LMI2} and pre- and post multiply with $[y^\top,z^\top]^\top$ with $z \in  \mathbb{C}^m$ arbitrary.
Using $V = PB-Q^\top$, with $Q:= (D-\lambda I)^\top C$ and $B^\top Py=0$, we find
\begin{equation}
    -y^H Q^H W^{-1}Qy - 2y^H Q^\top z -z^H Wz +r^2 z^H z \leq 0.
\end{equation}
Choose $z = -W^{-1}Qy$ such that the above inequality reduces to  $r^2 z^H z \leq 0$, which implies $z=0$. Since the matrix $D-\lambda I$ is invertible, it follows that $z=0 \Leftrightarrow Cy=0$. In summary: $y^H(X^\top P+PX)y=0 \implies Cy = 0$.

Next, let $X$ have non-zero eigenvalues, and let \eqref{eq:XP} hold. Assume $P\succeq 0$. Let $v\neq 0$ be an eigenvector of $X$ such that $Xv = av$ with  $\textup{Re}\{a\}>0$. Then, we have
\begin{equation}
    v^H Pv \geq 0\:\:\textup{  and  }\:\: 2\textup{Re}\{a\}v^H Pv \leq 0,
\end{equation}
where the second inequality follows from \eqref{eq:XP}. Since $\textup{Re}\{a\}>0$ it follows that $v^H Pv = 0$. This implies $v^H(X^\top P+PX)v = 0$, and thus, as shown previously $Cv=0$. Recall that $X = A-B(D-\lambda I)^{-1}C$. Then $Xv = Av = av$, which implies that $v$ is an eigenvector of $A$ corresponding to the unstable eigenvalue $a$. Since $A$ is Hurwitz this is a contradiction and our assumptions must be false: either \eqref{eq:LMI2} does not have a solution, or if \eqref{eq:LMI2} admits a solution, $P\succeq 0$ can not be true.

Next, suppose $X$ has zero eigenvalues. Then, $Xv = 0$, and thus $v^H(X^\top P+PX)v = 0$, which implies $Cv = 0$. However, $Xv = 0 = Av$ contradicts the fact that $A$ is Hurwitz. Hence, in this case \eqref{eq:LMI2} can not admit a feasible solution. Thus, for $\lambda\in\Lambda$ the matrix $X$ is not Hurwitz, and any $P$ satisfying \eqref{eq:LMI2} cannot be positive semi-definite.

Next, we collect the $\Pi$ matrices for $\sigma=+1$, in the sets
\begin{subequations}
\begin{align}
    \Sigma^{+1}_{\succeq 0}(H)&:=\left\{\Pi \in \mathbf{\Pi}(H,\mathbb{R}\setminus \Lambda)\mid \sigma = +1 \right\} \\
    \Sigma^{+1}_{\times}(H)&:=\left\{\Pi \in \mathbf{\Pi}(H,\Lambda)\mid \sigma = +1, \right\},
    \end{align}
\end{subequations}
and define the corresponding regions
\begin{equation*}
     S_2(H) =\!\!\!\!\!\! \bigcap_{\Pi \in \Sigma_{\succeq 0}^{+1}(H)}  \!\!\!\!\!\!  \mathcal{S}(\Pi)\; \textup{ and } \;\;S_3(H)=\!\!\!\!\!\! \bigcap_{\Pi \in  \Sigma^{+1}_{\times}(H)} \!\!\!\!\!\! \mathcal{S}(\Pi). 
\end{equation*}
Since $\Sigma_{\succeq 0}^{-1}(H) \cup \Sigma_{\succeq 0}^{+1}(H)\cup \Sigma^{+1}_{\times}(H) = \mathbf{\Pi}(H,\mathbb{R})$, it follows from \eqref{eq:SG0} that
\begin{equation}
    \overline{\textup{SG}}(H)=S_1(H)\cap S_2(H)\cap S_3(H).
\end{equation}
Using De Morgan's law we then find $\overline{\textup{SG}}(H)^c = {S}_1(H)^c\cup {S}_2(H)^c\cup {S}_3(H)^c$. Clearly, $S_1(H)$ is a bounded set as it results from taking intersections, with at least one bounded disk with radius $\|H\|$. This directly implies that $S_1(H)^c$ is unbounded. Both ${S}_2(H)$ and ${S}_3(H)$ are unbounded sets, since  $\{z\in \mathbb{C} \mid p_0\leq\textup{Re}(z) \leq p_1 \}\subset S_2(H)$ and $\mathbb{R} \setminus \Lambda \subset S_3(H)$.  Furthermore,  $S_2(H)^c$ is unbounded, as it results from the union of disks centered at $\lambda \in \mathbb{R}\setminus \Lambda$, which is unbounded.  $S_3(H)^c$ results from the union of disks centered at $\lambda\in\Lambda$, which, since $\Lambda$ is bounded and $S_3(H)^c$ being bounded by the spectrum $\rho(H)$ implies that $S_3(H)^c$ is bounded. Thus, $( \overline{\textup{SG}}(H)^c)_\infty = {S}_1(H)^c \cup S_2(H)^c$. We then find
\begin{equation*}
\begin{split}
    \operatorname{patch}(\overline{\textup{SG}}(H)) &=(S_1(H)\cap S_2(H)\cap S_3(H)) \cup S_3(H)^c\\
& = S_1(H)\cap S_2(H),
    \end{split}
\end{equation*}

where we used $S_3(H)^c \subseteq S_1(H)\cap S_2(H)$, which is evident from the construction of the regions. Since $\Sigma_{\succeq 0}^{-1}(H) \cup \Sigma_{\succeq 0}^{+1}(H)  = \mathbf{\Pi}_{\succeq 0}(H,\mathbb{R}),$ it follows that
\begin{equation}\label{eq:Patch}
    \operatorname{patch}(\overline{\textup{SG}}(H))= S_1(H)\cap S_2(H) =\!\!\!\! \bigcap_{\Pi \in \mathbf{\Pi}_{\succeq 0}(H,\mathbb{R})} \!\!\!\mathcal{S}(\Pi).
\end{equation}
Finally, we use \cite[Theorem 15]{Groot25} and $H=\mathcal{R}_{BLS}$ to conclude 
\begin{equation*}
    \textup{SG}(\mathcal{R}) \subseteq \!\!\! \!\!\!\! \bigcap_{\Pi \in \mathbf{\Pi}_{\succeq 0}(\mathcal{R}_{\textup{BLS}},\mathbb{R})} \!\!\!\!\!\!\mathcal{S}(\Pi) \quad \textup{and} \quad \textup{SG}_e(\mathcal{R}) \subseteq \!\!\! \!\!\!\! \bigcap_{\Pi \in \mathbf{\Pi}_{\succeq 0}(\mathcal{R}_{\textup{BLS}},\mathbb{R})}\!\!\!\!\!\!\mathcal{S}(\Pi).
\end{equation*}
Applying the equality in \eqref{eq:Patch} leads to the result.

\end{document}